\documentclass[pre,longbibliography]{revtex4-1}

\usepackage{graphicx}
\usepackage{dcolumn}
\usepackage{bm}

\usepackage[utf8]{inputenc}
\usepackage[T1]{fontenc}
\usepackage{mathptmx}
\usepackage{etoolbox}
\usepackage{color}

\usepackage{algorithm}
\usepackage[algo2e]{algorithm2e}
\usepackage{amssymb}
\usepackage{amsthm}
\usepackage{amscd}
\usepackage{amsmath}
\usepackage{amsfonts}
\usepackage{mathrsfs}
\usepackage{subfigure} 
\usepackage{float,multirow,graphics,fancyhdr}
\usepackage{geometry}
\usepackage{tabularx}  
\usepackage{booktabs}
\usepackage[table]{xcolor}
\definecolor{BestColor}{RGB}{255,50,50}   
\definecolor{SecondColor}{RGB}{255,255,0} 
\definecolor{ThirdColor}{RGB}{50,205,50} 

\newtheorem{theorem}{Theorem}[section]
\newtheorem{remark}[theorem]{Remark}

\newtheorem{proposition}[theorem]{Proposition}
\newtheorem{definition}{Definition}[section]
\newtheorem{lemma}[theorem]{Lemma}

\makeatletter
\def\@email#1#2{%
 \endgroup
 \patchcmd{\titleblock@produce}
  {\frontmatter@RRAPformat}
  {\frontmatter@RRAPformat{\produce@RRAP{*#1\href{mailto:#2}{#2}}}\frontmatter@RRAPformat}
  {}{}
}%
\makeatother
\begin{document}

\title[Community detection of hypergraphs by Ricci flow]{Community detection of hypergraphs by Ricci flow}
\author{Yulu Tian}
\affiliation{ 
School of Mathematical Sciences, Key Laboratory of Mathematics and Complex Systems of MOE,\\
Beijing Normal University, Beijing, 100875, China
}%
\author{Jicheng Ma}%
\affiliation{ 
School of Mathematics, Renmin University of China, Beijing, 100872, China
}%
\author{Yunyan Yang}%
\affiliation{ 
School of Mathematics, Renmin University of China, Beijing, 100872, China
}%
\author{Liang Zhao}
\altaffiliation{%
	Corresponding author
}
\affiliation{%
School of Mathematical Sciences, Key Laboratory of Mathematics and Complex Systems of MOE,\\
Beijing Normal University, Beijing, 100875, China
}%
\email{liangzhao@bnu.edu.cn}

\date{\today}

\begin{abstract}
Community detection in hypergraphs is both instrumental for functional module identification and intricate due to higher-order interactions among nodes. We define a hypergraph Ricci flow that directly operates on higher-order interactions of hypergraphs and prove long-time existence of the flow. Building on this theoretical foundation, we develop HyperRCD-a Ricci-flow-based community detection approach that deforms hyperedge weights through curvature-driven evolution, which provides an effective mathematical representation of higher-order interactions mediated by weighted hyperedges between nodes. Extensive experiments on both synthetic and real-world hypergraphs demonstrate that HyperRCD exhibits remarkable enhanced robustness to topological variations and competitive performance across diverse datasets.
\end{abstract}

\maketitle

\section{Introduction}

Graphs (directed or undirected) have served as powerful mathematical models for real-world networks, demonstrating remarkable success in both theoretical and applied domains through graph-theoretic analysis. However, many real-world systems involve higher-order interactions that transcend pairwise relationships between nodes. Hypergraphs naturally generalize graph theory by permitting edges (called hyperedges) to connect any number of vertices, thus providing a more appropriate framework for modeling such higher-order interactions. This mathematical structure has proven particularly valuable across scientific and engineering disciplines. For example, directed hypergraphs effectively model chemical reaction networks \cite{Jost2019hypergraph}, where vertices represent chemical elements and hyperedges correspond to reactions involving multiple elements. Similarly, undirected hypergraphs can capture structures of social networks \cite{Yu2021self,Yu2023self}, with vertices representing individuals and hyperedges encoding group relationships. These applications demonstrate how hypergraph theory offers a rigorous foundation for analyzing complex systems with multi-way interactions.

From a mathematical perspective, the geometric structures of graphs and hypergraphs, particularly the notion of curvature, play a fundamental role. This viewpoint has gained significant attention in recent years, yielding substantial theoretical advances and applications. In \cite{Ollivier2009ricci}, Ollivier defined a notion of Ricci curvature on metric spaces, including graphs, equipped with a Markov chain. Lin and Yau presented a generalization of the lower Ricci curvature bound within the framework of undirected graphs in \cite{Lin2010Ricci}, and Lin-Lu-Yau \cite{Lin2011Ricci} studied a modified definition of Ollivier’s Ricci curvature on graphs. 
In \cite{Bai2021sum}, Bai, Huang, Lu and Yau utilized a limit-free version of curvature called star coupling Ricci curvature to reduce the number of free parameters. For further studies of Ollivier’s Ricci curvature in the undirected graph setting, we refer to \cite{,Bauer2012ollivier,Benson2021volume,Devriendt2022discrete,Jost2014ollivier,Munch2019ollivier} and the references therein. For directed graphs, Yamada \cite{Yamada2019ricci} first proposed a generalization of Lin-Lu-Yau Ricci curvature, calculating it for several specific examples and providing various estimates. Ozawa, Sakurai, and Yamada \cite{Ozawa2020geometric} introduced a new generalization for strongly connected directed graphs by utilizing the mean transition probability kernel, which is involved in the formulation of the Chung Laplacian. 

Following graph curvature theories, there have been growing efforts to develop analogous tools for the Ricci curvature of hypergraphs in recent years. Eidi and Jost introduced Ollivier-Ricci curvature of directed hypergraphs in \cite{Eidi2020ollivier}, which was used to analyze certain chemical reaction models in \cite{Eidi2020edge}. Leal et al.  \cite{Leal2021forman}, defined Forman–Ricci curvature for directed and undirected hypergraphs, recovering the curvature for graphs as a special case, and determined the upper and lower bounds of Forman-Ricci curvature both for hypergraphs and their special case of graphs. Leal, Eidi and Jost \cite{Leal2020ricci} introduced a random model to generate directed hypergraphs and studied properties such as the degree of nodes and edge curvature using numerical simulations. By applying hypergraph shuffling to the metabolic network of \textit{E. coli}, they demonstrated that changes in the wiring of a hypergraph can be detected by Forman-Ricci and Ollivier-Ricci curvatures. In \cite{Coupette2022ollivier}, the authors developed ORCHID (Ollivier Ricci Curvature for Hypergraphs In Data) of undirected, unweighted hypergraphs based on random walks on the nodes, which provided a flexible framework generalizing Ollivier-Ricci curvature to hypergraphs.

The Ricci flow
\begin{equation}\label{Ricci flow}
	\left\{
	\begin{aligned}
		&\  \frac{\partial}{\partial t}g=-2\text{Ric}(g),\\
		&\  g(0)=g_0,
	\end{aligned}
	\right.
\end{equation}
introduced by Hamilton in 1982 \cite{Hamilton1982three}, is a fundamental tool in geometric analysis that evolves Riemannian metrics through their Ricci curvature. This method has proven exceptionally powerful for extracting key geometric and topological properties of manifolds, most notably in Perelman's proof of the three-dimensional Poincaré conjecture \cite{Perelman2002entropy,Perelman2003finite,Perelman2003ricci}. As discrete analogues of their continuum counterparts, Ricci curvature and Ricci flow on graphs have garnered significant research interest, evolving into both a rigorous theoretical framework and a versatile tool with multidisciplinary applications \cite{Jin2008discrete,Sarkar2009greedy,Shepherd2022feature}. Among various applications, the innovative approach of community detection via discrete Ricci flow has attracted significant research attention. Here, community detection refers to the identification of densely interconnected vertex subsets with sparse connections between groups. For comprehensive background and alternative methods in community detection of graphs, we refer readers to \cite{Camacho2020four, Fortunato2010community, Fang2020survey, Leskovec2010empirical, Li2020community, Rossetti2018community, Su2022comprehensive}. In \cite{Ni2019community}, Ni et al. found that discrete Ollivier-Ricci flow can be used to detect community structures of networks. Sia et al. \cite{Sia2019ollivier} constructed a community detection algorithm by progressively removing negatively curved edges. In \cite{Lai2022normalized}, Lai et al. leveraged a normalized flow based on discrete Ricci curvature, to deform a graph, and as a result, intracommunity nodes became closer while intercommunity nodes dispersed. Recently, Bai et al. \cite{Bai2024ollivier} studied a Ricci flow based on Lin-Yu-Yau Ricci curvature and investigated the existence and uniqueness properties of its solutions. In \cite{Cushing2023bakry,Cushing2025bakry}, the authors introduced a Bakry-{\'E}mery Ricci flow on weighted graphs, which preserves the Markovian property, and the graphs turn out to be curvature-sharp as time goes to infinity. Later, Hua, Lin and Wang \cite{Hua2024version} studied the local existence and uniqueness, and obtained results on long-time convergence or finite time blow-up for the Bakry-{\'E}mery Ricci flow on finite trees and circles.

Community detection in hypergraphs remains an active research area. The inherent complexity of higher-order interactions makes this problem particularly challenging. Recent methodological advances include optimal statistical limit \cite{Chien2018community}, regularized tensor power iteration \cite{Ke2019community},random walk approaches \cite{Carletti2021random}, a probabilistic framework \cite{Ruggeri2023community}, and mutual information maximization \cite{Kritschgau2024community}, etc. On the other hand, the strong geometric representational power of Ricci curvature makes it particularly well-suited for analyzing hypergraph structures. For example, \cite{Hacquard2024hypergraph} proposed an edge transport perspective of hypergraph community detection, where Ricci curvature of hyperedges serves as a geometric criterion.

In light of the aforementioned work, the key contributions of this paper are : (i) Construct a hypergraph Ricci flow framework with necessary theoretical foundations. (ii) Develop a hypergraph community detection algorithm that directly operates on higher-order interactions of hypergraphs, bypassing graph reduction techniques (e.g., clique expansion or star expansion) that may incur information loss. (iii) Validate the method's effectiveness through experiments on synthetic and real-world networks.


\section{Ollivier–Ricci curvature of hypergraphs}\label{basic}

Let $\Gamma=(V,H,\bm{w})$ be a weighted undirected hypergraph, where $V=(x_1,x_2,\cdots,x_n)$ is vertex set, $H=(h_1,h_2,\cdots,h_m)$ is the hyperedge set, and $\bm{w}=(w_{h_1},w_{h_2},\cdots,w_{h_m})\in\mathbb{R}^m_+$ is the vector of weights on hyperedges. When referring to generic vertices without specification, we omit subscripts and denote nodes by $u$ and $v$. For two vertices $u, v\in V$, if there exists a hyperedge $h_l\in H$ such that both of them are contained within $h_l$, we say that $u$ and $v$ are adjacent and denote this relationship as $u\sim v$. Furthermore, if there is a sequence of hyperedges $\{h_{l_1},h_{l_2},\cdots,h_{l_s}\}$ ($s\leq m$) such that $u\in h_{l_1}$, $v\in h_{l_s}$ and $h_{l_k}\cap h_{l_{k+1}}\neq\emptyset$ for $1\leq k\leq s-1$, then vertices $u$ and $v$ are connected by the hyperpath $\gamma=\{h_{l_1},h_{l_2},\cdots,h_{l_s}\}$. If each pair of two different vertices in a hypergraph $\Gamma$ can be connected by a hyperpath, we call $\Gamma$ a connected hypergraph. Unless otherwise specified, any mention of hypergraphs refers to undirected, connected hypergraphs so that we can provide a definition of distance. The distance between two vertices $u, v\in V$ is given by
\begin{equation}\label{distance def}
	d(u,v):=\inf_\gamma\sum_{h\in\gamma}w_h,
\end{equation}
where the infimum is taken over all paths $\gamma$ connecting $u$ and $v$. Since hyperedges can connect multiple vertices, we need to establish a definition of length for a hyperedge $h_l=\{x_{l_1},x_{l_2},\cdots,x_{l_s}\}\in H$ as follow
\begin{equation}\label{distance hyperedge}
	d(h):=\sum_{1\leq i<j\leq s}d(x_{l_i},x_{l_j}).
\end{equation}
For the case $k=2$, which indicates that the (hyper)edge contains only two vertices, the length of $h=\{u,v\}$ is $d(h)=d(u,v)$. By \cite[Lemma 2.1.]{Ma2024modified}, we have the following lemma concerning distance.
\begin{lemma}\label{distence lem}
	For any two fixed vertices $u, v\in V$, the distance $d(u,v)$ is locally Lipschitz in $\mathbb{R}^m_+$ with respect to $\bm{w}=(w_{h_1},w_{h_2},\cdots,w_{h_m})$, i.e., for any $\bm{w},  \widetilde{\bm{w}}\in\mathbb{R}^m_+$, if $d$ and $\widetilde{d}$ are two distance functions determined by $\bm{w}$ and
	$\widetilde{\bm{w}}$ respectively, then there holds
	\begin{equation*}
		\vert d(u,v)-\widetilde{d}(u,v)\vert\leq\sqrt{m}\vert\bm{w}-\widetilde{\bm{w}}\vert.
	\end{equation*}
\end{lemma}
As the foundation for developing curvature flows on hypergraphs, we first present a rigorous definition of the Ollivier-Ricci curvature of hyperedges, which will serve as the key geometric quantity in subsequent evolutionary equations.
\begin{definition}\label{curvature def}
	For a weighted hypergraph $\Gamma=(V,H,\bm{w})$, $\forall h_l=\{x_{l_1},x_{l_2},\cdots,x_{l_s}\}\in H$, we define the Ollivier-Ricci curvature of the hyperedge $h_l$ as
	\begin{equation}\label{curvature eq}
		\kappa_\alpha(h_l):=1-\frac{W_{h_l}}{d(h_l)},
	\end{equation}
	where $\alpha\in[0, 1]$, the probability measure $\mu_{x_{l_i}}^\alpha$ of $x_{l_i}\in h_l$ is defined on $V$ as follows
	\begin{equation}\label{measure eq}
		\mu_{x_{l_i}}^\alpha(z):=\left\{
		\begin{aligned}
			&\  \alpha, &&\text{if } z=x_{l_i},\\
			&\  (1-\alpha)\sum_{h': x_{l_i},z\in h'}\frac{1}{\vert h'\vert-1}\frac{w_{h'}}{\sum\limits_{h'':x_{l_i}\in h''}w_{h''}}, &&\text{if } z\neq x_{l_i} \text{ and } z\sim x_{l_i},\\
			&\  0, &&\text{otherwise},
		\end{aligned} \qquad \forall z\in V,
		\right.
	\end{equation}
	$W_{h_l}:=\sum_{1\leq i<j\leq s}W(\mu_{x_{l_i}}^\alpha,\mu_{x_{l_j}}^\alpha)$ and $W(\mu_{x_{l_i}}^\alpha,\mu_{x_{l_j}}^\alpha)$ is the $1$-Wasserstein distance between these two discrete measures defined as follows
	\begin{equation}\label{Wasserstein def}
		W(\mu_{x_{l_i}}^\alpha,\mu_{x_{l_j}}^\alpha):=\inf_{A}\sum_{u,v\in V}A(u,v)d(u,v),
	\end{equation}
	where $A: V\times V\to [0,1]$ runs over all maps satisfying
	\begin{equation*}
		\left\{
		\begin{aligned}
			&\  \sum_{v\in V}A(u,v)=\mu_{x_{l_i}}^\alpha(u),\\
			&\  \sum_{u\in V}A(u,v)=\mu_{x_{l_j}}^\alpha(v).
		\end{aligned}
		\right.
	\end{equation*}
\end{definition}

\begin{remark}\label{graph Lip re}
	For a weighted hypergraph $\Gamma=(V,H,\bm{w})$, we note that:
	
	(i) If $k=2$, i.e., there is a hyperedge has exactly two vertices $h=\{u,v\}$, the Ollivier-Ricci curvature (\ref{curvature eq}) formally coincides with the graph edge curvature as follows
	\begin{equation*}
		\kappa_\alpha(h):=1-\frac{W(\mu_u^\alpha,\mu_v^\alpha)}{d(u,v)}.
	\end{equation*}
	However, the probability measure (\ref{measure eq}) may differ from those in the graph case unless all hyperedges contain exactly two vertices.
	
	(ii) By Definition \ref{curvature def} and \cite{Ma2024evolution}, it is clear that the probability measure $\mu_{x_i}^\alpha$ is locally Lipschitz in $\mathbb{R}^m_+$ with respect to $\bm{w}$ for any $x_i\in V$.
\end{remark}
We next introduce the concept of $1$-Lipschitz functions on hypergraphs and a fundamental principle in optimal transport theory, namely the Kantorovich-Rubinstein duality.
\begin{definition}
	Let $\Gamma=(V,H,\bm{w})$ be a weighted hypergraph. A function $f: V\to\mathbb{R}$ is said to be $1$-Lipschitz if
	\begin{equation*}
		\vert f(u)-f(v)\vert\leq d(u, v),
	\end{equation*}
	for all $u,v\in V$. The set of all $1$-Lipschitz functions on $V$ is denoted by $1$-Lip.
\end{definition}

\begin{proposition}\label{Kantorovich duality}
	\emph{(Kantorovich-Rubinstein duality, \cite{Bourne2018ollivier, Villani2021topics})} Let $\Gamma=(V,H,\bm{w})$ be a weighted hypergraph and  $\mu_{x_i}^\alpha, \mu_{x_j}^\alpha$ be two probability measures on $V$. Then
	\begin{equation*}
		W(\mu_{x_i}^\alpha,\mu_{x_j}^\alpha)=\sup_{\phi\in 1\text{-Lip}}\sum_{z\in V}\phi(z)\left(\mu_{x_i}^\alpha(z)-\mu_{x_j}^\alpha(z)\right).
	\end{equation*}
\end{proposition}


\section{Ricci flow on hypergraphs}\label{result}

In this section, we study a flow governing the evolution of weight $\bm{w}=(w_{h_1},w_{h_2},\cdots,w_{h_m})$ over time. To begin with, for any $h\in H$, we note that since both $W_h$ and $d_h$ depend on $\bm{w}$, they can be treated as functions of $\bm{w}$, denoted by $W_h(\bm{w})$ and $d_h(\bm{w})$, respectively. The Ricci flow on hypergraphs is
\begin{equation}\label{flow eq}
	\left\{
	\begin{aligned}
		&\  w_{h_l}'(t)=-d(h_l)\kappa_\alpha(h_l),\\
		&\  h_l=\{x_{l_1},x_{l_2},\cdots,x_{l_s}\}\in H,\\
		&\  w_{h_l}(0)=w_{0,l},
	\end{aligned}
	\right.
\end{equation}
where $t\geq0$ and $l=1,2,\cdots,m$. The right-hand side of the above Ricci flow depends on $\bm{w}$, specifically given by
$$-d(h_l)\kappa_\alpha(h_l)=W_{h_l}(\bm{w})-d(h_l)(\bm{w})=\sum_{1\leq i<j\leq s}\left(W(\mu_{x_{l_i}}^\alpha, \mu_{x_{l_j}}^\alpha)(\bm{w})-d(x_{l_i},x_{l_j})(\bm{w})\right).$$
To prove the long-time existence of the flow, we first introduce the following lemma.
\begin{lemma}\label{Wasserstein lem}
	Let $\Gamma=(V,H,\bm{w})$ be a weighted hypergraph and $h_l=\{x_{l_1},x_{l_2},\cdots,x_{l_s}\}\in H$. For any two different vertices $x_{l_i}, x_{l_j} \in h_l$, the $1$-Wasserstein distance $W(\mu_{x_{l_i}}^\alpha, \mu_{x_{l_j}}^\alpha)(\bm{w})$ is locally Lipschitz in $\mathbb{R}^m_+$ with respect to $\bm{w}$. In addition, we have that $W_{h_l}(\bm{w})-d(h_l)(\bm{w})$ is also locally Lipschitz in $\mathbb{R}^m_+$ with respect to $\bm{w}$.
\end{lemma}

\begin{proof}
	For any two fixed vertices $x_{l_i}, x_{l_j} \in h_l$ and two vector of weights  $\bm{w}=(w_{h_1},w_{h_2},\cdots,w_{h_m})$, $\widetilde{\bm{w}}=(\widetilde{w}_{h_1},\widetilde{w}_{h_2},\cdots,\widetilde{w}_{h_m})$ in $\mathbb{R}^m_+$, assume that $W(\mu_{x_{l_i}}^\alpha, \mu_{x_{l_j}}^\alpha)$ and $W(\widetilde{\mu}_{x_{l_i}}^\alpha, \widetilde{\mu}_{x_{l_j}}^\alpha)$ are the Wasserstein distances determined by $\bm{w}$ and $\widetilde{\bm{w}}$, respectively, as well as the probability measures $\mu, \widetilde{\mu}$ and the distances $d,\widetilde{d}$. Since the cardinality of the hyperedge set $H$ is finite, there exist some constants $M>0$ and $\delta>0$ such that for each $l=1,2,\cdots,m$, we have
	\begin{equation*}
		M^{-1}\leq w_{h_l}\leq M,\  M^{-1}\leq \widetilde{w}_{h_l}\leq M,\  \vert w_{h_l}-\widetilde{w}_{h_l}\vert\leq\delta^{-1}.
	\end{equation*}
	On the other hand, 	
	since the cardinality of the vertices set $V$ is finite, by Proposition \ref{Kantorovich duality}, there exists some $f\in1\text{-Lip}$ such that
	\begin{equation*}
		W(\mu_{x_{l_i}}^\alpha, \mu_{x_{l_j}}^\alpha)(\bm{w})=\sum_{z\in V}f(z)\left(\mu_{x_{l_i}}^\alpha(z,\bm{w})-\mu_{x_{l_j}}^\alpha(z,\bm{w})\right).
	\end{equation*}
	It follows from Lemma \ref{distence lem} that for any $u,v\in V$, we have
	\begin{equation*}
		\frac{\widetilde{d}(u,v)}{d(u,v)-\widetilde{d}(u,v)}\leq\sqrt{m}M\vert\bm{w}-\widetilde{\bm{w}}\vert.
	\end{equation*}
	Therefore, the function
	\begin{equation*}
		\widetilde{f}(u)=\frac{f(u)}{1+\sqrt{m}M\vert\bm{w}-\widetilde{\bm{w}}\vert}, \quad \forall u\in V,
	\end{equation*}
	satisfies
	\begin{equation*}
		\vert \widetilde{f}(u)-\widetilde{f}(v)\vert=\frac{\vert f(u)-f(v)\vert}{1+\sqrt{m}M\vert\bm{w}-\widetilde{\bm{w}}\vert}\leq\widetilde{d}(u,v).
	\end{equation*}
	i.e. $\widetilde{f}\in1\text{-}\widetilde{\text{Lip}}$.
	By Remark \ref{graph Lip re}(ii), there exists a constant $C>0$, which only depends on the constants $M$ and $\delta$, such that
	\begin{equation*}
		\vert\mu_{x_{l_i}}^\alpha(u,\bm{w})-\widetilde{\mu}_{x_{l_i}}^\alpha(u,\widetilde{\bm{w}})\vert\leq C\vert\bm{w}-\widetilde{\bm{w}}\vert, \quad \forall x_{l_i}\in h_l, u\in V.
	\end{equation*}
	Without loss of generality, assume that $W(\mu_{x_{l_i}}^\alpha, \mu_{x_{l_j}}^\alpha)(\bm{w})\geq W(\widetilde{\mu}_{x_{l_i}}^\alpha, \widetilde{\mu}_{x_{l_j}}^\alpha)(\widetilde{\bm{w}})$, and we have
	\begin{equation*}
		\begin{aligned}
			&W(\mu_{x_{l_i}}^\alpha, \mu_{x_{l_j}}^\alpha)(\bm{w})- W(\widetilde{\mu}_{x_{l_i}}^\alpha, \widetilde{\mu}_{x_{l_j}}^\alpha)(\widetilde{\bm{w}})\\
			\leq\ &\sum_{z\in V}f(z)\left(\mu_{x_{l_i}}^\alpha(z,\bm{w})-\mu_{x_{l_j}}^\alpha(z,\bm{w})\right)-\sum_{z\in V}\widetilde{f}(z)\left(\widetilde{\mu}_{x_{l_i}}^\alpha(z,\widetilde{\bm{w}})-\widetilde{\mu}_{x_{l_j}}^\alpha(z,\widetilde{\bm{w}})\right)\\
			\leq\ &\sum_{z\in V}\vert f(z)\vert\left(\vert\mu_{x_{l_i}}^\alpha(z,\bm{w})-\widetilde{\mu}_{x_{l_i}}^\alpha(z,\widetilde{\bm{w}})\vert+\vert\mu_{x_{l_j}}^\alpha(z,\bm{w})-\widetilde{\mu}_{x_{l_j}}^\alpha(z,\widetilde{\bm{w}})\vert\right)\\
			&+\sum_{z\in V}\vert f(z)-\widetilde{f}(z)\vert\vert\widetilde{\mu}_{x_{l_i}}^\alpha(z,\widetilde{\bm{w}})-\widetilde{\mu}_{x_{l_j}}^\alpha(z,\widetilde{\bm{w}})\vert\\
			\leq\ &\sum_{z\in V}2C\vert\bm{w}-\widetilde{\bm{w}}\vert\vert f(z)\vert+\sum_{z\in V}\frac{\sqrt{m}M\vert\bm{w}-\widetilde{\bm{w}}\vert}{1+\sqrt{m}M\vert\bm{w}-\widetilde{\bm{w}}\vert}\vert f(z)\vert\left(\widetilde{\mu}_{x_{l_i}}^\alpha(z,\widetilde{\bm{w}})+\widetilde{\mu}_{x_{l_j}}^\alpha(z,\widetilde{\bm{w}})\right)\\
			\leq\ &2n(C+\sqrt{m}M)\Vert f\Vert_{L^\infty(V)}\vert\bm{w}-\widetilde{\bm{w}}\vert.
		\end{aligned}
	\end{equation*}
	This completes the proof of the Lemma.
\end{proof}

Next, we present the main theorem of this section, which guarantees the long-time existence of solutions to the flow (\ref{flow eq}).
\begin{theorem}
	Let $\Gamma=(V,H,\bm{w}_0)$ be a weighted hypergraph with an initial weight $\bm{w}_0\in \mathbb{R}^m_+$. Then there exists a unique solution $\bm{w}(t)$ of the flow (\ref{flow eq}) for all time $t\in[0,\infty)$, where $\bm{w_0}=(w_{0,1},w_{0,2},\cdots,w_{0,m})$, and $\bm{w}(t)=(w_{h_1}(t),w_{h_2}(t),\cdots,w_{h_m}(t))$.
\end{theorem}

\begin{proof}
	We first note that for any $\bm{w}_0\in\mathbb{R}^m_+$, the evolution of weights (\ref{flow eq}) is a system of ordinary differential equations as follows
	\begin{equation}\label{system eq}
		\left\{
		\begin{aligned}
			&\  \bm{w}'(t)=\bm{W}_H(t)-\bm{d}_H(t),\\
			&\  \bm{w}(0)=\bm{w}_0,
		\end{aligned}
		\right.
	\end{equation}
	where for simplicity, we denote $\left(W_{h_1}(t)-d(h_1)(t),W_{h_2}(t)-d(h_2)(t),\cdots,W_{h_m}(t)-d(h_m)(t)\right)$ by $\bm{W}_H(t)-\bm{d}_H(t)$. By Lemma \ref{Wasserstein lem}, it is clear that $\bm{W}_H-\bm{d}_H$ is locally Lipschitz with respect to $\bm{w}$. Through the classical theorem on the existence and uniqueness of solutions to ordinary differential systems, there exists a constant $T>0$ such that the system (\ref{system eq}) has a unique solution $\bm{w}(t)$ on $[0,T]$.
	
	Let
	\begin{equation*}
		T^*:=\sup\left\{T>0: (\ref{system eq}) \text{ has a unique solution on } [0,T]\right\}
	\end{equation*}
	and define
	\begin{equation*}
		\psi(t):=\min\{w_{h_1}(t),w_{h_2}(t),\cdots,w_{h_m}(t)\}, \quad \Psi(t):=\max\{w_{h_1}(t),w_{h_2}(t),\cdots,w_{h_m}(t)\}.
	\end{equation*}
	Next, we prove that $T^*$ must be infinite. Suppose not, then according to the ODE theory, we have
	\begin{equation}\label{contradiction eq}
		\liminf\limits_{t\to T^*}\psi(t)=0, \quad\text{or}\quad \limsup\limits_{t\to T^*}\psi(t)=\infty.
	\end{equation}
	For any hyperedge $h_l=(x_{l_1},x_{l_2},\cdots,x_{l_s})\in H$, by the definition (\ref{distance hyperedge}), the system (\ref{system eq}) can be rewritten as follows.
	\begin{equation*}
		\begin{aligned}
			w_{h_l}'(t)&=\sum_{1\leq i<j\leq s}\left(W(\mu_{x_{l_i}}^\alpha, \mu_{x_{l_j}}^\alpha)-d(x_{l_i},x_{l_j})\right)\\
			&\geq-\sum_{1\leq i<j\leq s}d(x_{l_i},x_{l_j})\\
			&\geq-\frac{s(s-1)}{2}w_{h_l}(t).
		\end{aligned}
	\end{equation*}
	Thus, we have
	\begin{equation}\label{weight leq}
		w_{h_l}(t)\geq w_{h_l}(0)e^{-\frac{s(s-1)}{2}T^*}, \quad \forall t\in[0,T^*).
	\end{equation}
	
	On the other hand, by definitions (\ref{distance def}) and (\ref{Wasserstein def}), for the two probability measures $\mu_{x_{l_i}}^\alpha$ and $\mu_{x_{l_j}}^\alpha$, we have
	\begin{equation}\label{Ww neq}
		\begin{aligned}
			W(\mu_{x_{l_i}}^\alpha, \mu_{x_{l_j}}^\alpha)&=\inf_{A}\sum_{u,v\in V}A(u,v)d(u,v)\\
			&\leq\left(\sum_{u,v\in V}B(u,v)\right)\left(\sum_{h\in\gamma}w_h\right)\\
			&=\sum_{h\in H}w_h,
		\end{aligned}
	\end{equation}
	where $B(u,v)$ is a coupling between $\mu_{x_i}^\alpha$ and $\mu_{x_j}^\alpha$ satisfying $\sum_{u,v\in V}B(u,v)=1$. It follows from (\ref{Ww neq}) that
	\begin{equation*}
		\begin{aligned}
			\frac{d}{dt}\sum_{h_l\in H}w_{h_l} &=\sum_{h_l\in H}\sum_{1\leq i<j\leq s}\left(W(\mu_{x_{l_i}}^\alpha, \mu_{x_{l_j}}^\alpha)-d(x_{l_i},x_{l_j})\right)\\
			&\leq\sum_{h_l\in H}\sum_{1\leq i<j\leq s}W(\mu_{x_{l_i}}^\alpha, \mu_{x_{l_j}}^\alpha)\\
			&\leq nm \sum_{h_l\in H}w_{h_l}.
		\end{aligned}
	\end{equation*}
	Hence, we have
	\begin{equation}\label{weight req}
		\sum_{h_l\in H}w_{h_l}(t)\leq e^{nmT^*}\sum_{h_l\in H}w_{h_l}(0), \quad \forall t\in[0,T^*).
	\end{equation}
	Combining (\ref{weight leq}) and (\ref{weight req}), we have
	\begin{equation}\label{contradiction weight}
		\psi(0)e^{-\frac{s(s-1)}{2}T^*}\leq\psi(t)\leq\Psi(t)\leq\sum_{h_l\in H}w_{h_l}(t)\leq e^{nmT^*}\sum_{h_l\in H}w_{h_l}(0), \quad \forall t\in[0,T^*),
	\end{equation}
	which contradicts (\ref{contradiction eq}). Throughout this paper, we use $s$ to denote a generic constant representing the cardinality of a hyperedge (i.e., the number of vertices it contains). Here in (\ref{contradiction weight}), it specifically refers to the vertex count of the minimally weighted hyperedge. Now the contradiction implies that $T^*$ can be extended to infinity, and consequently, we prove that there exists a unique solution $\bm{w}(t)$ of the flow (\ref{flow eq}) for all time $t\in[0,\infty)$.
\end{proof}


\section{Algorithm}\label{algorithm}

The discrete version of the flow (\ref{flow eq}) can be written as
\begin{equation}\label{discrete flow eq}
	\left\{
	\begin{aligned}
		&\  w_{h_l}(t_{k+1})- w_{h_l}(t_k)=\eta\left(W_{h_l}-d(h_l)\right)(t_k)=\eta\sum_{1\leq i<j\leq s}\left(W(\mu_{x_{l_i}}^\alpha, \mu_{x_{l_j}}^\alpha)-d(x_{l_i},x_{l_j})\right)(t_k),\\
		&\  h_l=\{x_{l_1},x_{l_2},\cdots,x_{l_k}\}\in H,\\
		&\  w_{h_l}(0)=w_{0,l},
	\end{aligned}
	\right. 
\end{equation}
where $\eta =(t_{k+1}-t_k)> 0$ is the step size of discretization. Smaller step sizes can usually improved progress of the flow through more iterations compared to using larger step sizes. To balance computational accuracy and efficiency, we set the step size $\eta= 0.1$. During this iteration process, edges with negative curvature are stretched, which increases the computational cost between different communities. By removing edges with large weights, we can identify communities effectively. This method reduces the weights of hyperedges between communities while preserving intra-community cohesiveness. The complete community detection algorithm (Ricci flow Community Detection of Hypergraphs, HyperRCD) is outlined by Algorithm 1 as follows.
\begin{algorithm}[H]
	\caption{HyperRCD}
	\label{one_algorithm}
	\KwIn{A weighted hypergraph $\Gamma=(V,H,\bm{w})$  ; Maximum iteration $K$ ; Step size $\eta$}
	\KwOut{Community detection result of $\Gamma$}
	
	\For{ $k = 1,\cdots ,K $}
	{\, Update hyperedge weights through (\ref{discrete flow eq})\\
		Let $\Gamma'=(V',H',\bm{w'})$ be the hypergraph after iteration\\
		\For{cutoff= $w'_{max},\cdots, w'_{min}$}
		{
			\For{ $ h'_l \in  H’ $}{
				\If{$w'_{h'_l}>{\it cutoff}$}{
					Remove the hyperedge $h'_l$
				}
			}
			Calculate the accuracy of community detection
		}
	} 
	Output the best community detection result of $\Gamma$
	
\end{algorithm}
The computational complexity of the proposed algorithm HyperRCD is primarily dominated by two factors: (i) the shortest path computation on the hypergraph, and (ii) solution of the linear programming problem. Let $E = \sum_{h \in H} \binom{s}{2}$, where $s$ denotes the cardinality of the hyperedge $h$. Let $D$ and $|V|$ represent the average degree and the number of vertices of the hypergraph respectively. The time complexities for the two tasks are $O(E\log|V|)$ and $O(E D^3)$ respectively. Here, despite the sparse connectivity of the hypergraph, where $D \ll E$, $O(E D^3)$ often surpasses $O(E \log|V|)$ in most scenarios. Consequently, the computational complexity of our approach is \(O(E D^3)\).


\section{Experiments}
\subsection{Datasets and comparison algorithms}
To assess the effectiveness of HyperRCD, we conducted extensive experiments on both synthetic and real-world networks (hypergraphs). The virtual datasets utilized in this study are generated using the Degree-Corrected Stochastic Block Model (DCSBM) \cite{Peixoto2014efficient}. This model is a robust statistical framework that allows for the creation of synthetic graphs with adjustable degree distributions and community structures. By specifying parameters related to the expected degrees of nodes and the probabilities of connections within and between communities, the DCSBM enables the generation of datasets that closely resemble the characteristics of real-world networks. 

For synthetic networks, we generated three series of networks of varying scales using the parameters listed in Table \ref{tab:parameters} to adjust the characteristics of the network. Denote the total hyperedge cardinality of the hypergraph $\Gamma$ by $\sum|H|$. For the D1 series, we maintain a network size of 100 nodes with 3 communities while systematically varying the average node degree from 3 to 30 (resulting in total hyperedge cardinality $\sum |H|$ ranging from 300 to 3000). This configuration evaluates algorithmic robustness under varying $\sum|H|$ and connectivity patterns. The D2 series fixes the network at 100 nodes with average node degree$=3$ and $\sum |H|=300$, but explores intra-community connectivity strength from 0.15 to 0.85. This manipulation creates increasingly “ambiguous” community boundaries when the connectivity strength is low, making it more difficult to separate communities. For the D3 series, we vary the the network size from 100 to 1000 nodes, while fixing the average node degree at 10, the number of communities at 10, and intra-community connectivity strength at 0.85. This configuration tests how well each algorithm copes with increasing network size. For complete parameter specifications, please refer to the provided accompanying source code, which is
available at https://github.com/mjc191812/Community-detection-of-hypergraphs-by-Ricci-flow. All these synthetic networks are used as input for the proposed algorithm to conduct the experiments and the detected results are compared with those of several popular algorithms, specifically Louvain\cite{blondel2008fast}, Girvan-Newman\cite{girvan2002community}, Label Propagation Algorithm (LPA)\cite{cordasco2010community}, Infomap\cite{rosvall2008maps}, and Walktrap\cite{pons2005computing}.

\begin{table}[htbp]
	\centering
	\caption{Parameter settings for synthetic data}
	\label{tab:parameters}
	\begin{tabular}{cccc}
		\toprule
		Parameter       & D1       & D2        & D3       \\
		\midrule
		\# Nodes    & 100      & 100       & 100-1000 \\
		\#Communities     & 3       & 3        & 10    \\
		Avg. node degree       & 3-30     & 3         & 10        \\
		$\sum|H|$ & 300-3000 & 300       & 1000-10000 \\
		Intra-community connectivity strength          & 0.85     & 0.15-0.85 & 0.85     \\
		\bottomrule
	\end{tabular}
\end{table}

A summary of dataset statistics for the real-world hypergraphs analyzed in this study, which originate from diverse applications, is provided in Table \ref{tab:2}. Following the experimental protocols established by Lee and Shin \cite{lee2023m} and Hacquard \cite{Hacquard2024hypergraph}, we evaluate our method on seven benchmark datasets, including Zoo and Mushroom \cite{Markelle2017uci}, Cora-C, Citeseer and Pubmed \cite{sen2008collective}, Cora-A \cite{rossi2015network}, NTU2012 \cite{chen2003visual}. We conduct a comprehensive evaluation of our method against advanced graph and hypergraph partitioning approaches. For graph-based counterparts, we include Node2vec \cite{grover2016node2vec} employing clique expansion with skip-gram modeling, DGI \cite{velivckovic2018deep} utilizing graph-level contrastive learning, and GRACE \cite{zhu2020deep} leveraging graph augmentation with contrastive representation learning. In hypergraph comparisons, we benchmark against S$^2$-HHGR \cite{zhang2021double} with spectral-spatial hybrid embeddings, TriCL \cite{lee2023m} implementing tri-directional contrastive learning for multi-scale structural consistency, and hypergraph modularity maximization algorithm \cite{kaminski2020community} optimizing community detection via modularity metrics. Our analysis extends the N-Ricci method \cite{Hacquard2024hypergraph}, which is a Ricci curvature-driven approach that applies node-based Ricci flow on clique-expanded hypergraphs. Crucially, since the clique expansion of hypergraphs is non-invertible, HyperRCD operates directly on the native hypergraph structure, better preserving the higher-order relationships inherent in hypergraph data compared to graph reduction approaches. The above comprehensive comparison spans diverse technical paradigms, from neural embeddings to curvature-driven dynamics, ensuring a rigorous validation of our method’s robustness and scalability, particularly for hypergraphs with large edges or heterogeneous community structures.
\begin{table}[h]
	\caption{\label{tab:2}Summary of real-world network characteristics}
	\centering
	\begin{tabular}{lccccccc}
		\toprule
		Dataset & Zoo & Mushroom  & Cora-C & Citeseer & Pubmed  & Cora-A  & NTU2012 \\
		\midrule
		\# Nodes & 101 & 8124 & 1434 & 1458 & 3840 & 2388 & 2012 \\
		\# Hyperedges & 43 & 298 & 1579 & 1079 & 7963 & 1072 & 2012 \\
		Avg. hyperedge size & 39.9 & 136.3 & 3.0 & 3.2 & 4.4 & 4.3 & 5.0 \\
		Avg. node degree & 17.0 & 5.0 & 3.3 & 2.4 & 9.0 & 1.9 & 5.0 \\
		\# Communities & 7 & 2 & 7 & 6 & 3 & 7 & 67 \\
		\bottomrule
	\end{tabular}
\end{table}
\subsection{Experimental Setup}

\textbf{Evaluation Metrics.}  
Various metrics exist to evaluate community detection algorithms. In this work, we employ a widely used metric, say the Normalized Mutual Information (NMI) \cite{ana2003robust} to evaluate the quality of the detected communities. NMI measures the agreement between two network partitions (e.g., predicted and ground truth), which is defined as:  
\[
\text{NMI} = \frac{-2 \sum_{a=1}^A \sum_{b=1}^B \mu_{ab} \log \left( \frac{\mu_{ab} \cdot \mathcal{N}}{\alpha_a \cdot \beta_b} \right)}{\sum_{a=1}^A \alpha_a \log \left( \frac{\alpha_a}{\mathcal{N}} \right) + \sum_{b=1}^B \beta_b \log \left( \frac{\beta_b}{\mathcal{N}} \right)},
\]  
where \( \mu_{ab} \) denotes the number of vertices co-assigned to group \( a \) in partition \( \mathcal{X} \) and group \( b \) in partition \( \mathcal{Y} \), \( \alpha_a \) and \( \beta_b \) represent group sizes of partitions \( \mathcal{X} \) and \( \mathcal{Y} \), and \( \mathcal{N} \) is the total number of vertices. This formulation normalizes mutual information by the combined entropy of both partitions, yielding values in \([0, 1]\). A score of 1 indicates perfect alignment, while 0 reflects independence.

NMI offers distinct advantages in community detection tasks. Its scale invariance and resistance to label permutations enable cross-algorithm comparison. Unlike modularity \cite{newman2004finding}, which depends on network density, NMI accounts for both structural correspondence and probabilistic node assignments. This robustness to hierarchical nesting and unequal group sizes makes it ideal for evaluating real-world networks with complex community structures \cite{lancichinetti2009community}. 

\textbf{Implementation Details.} It is important to note that the Ollivier–Ricci curvature of hypergraphs depends on the selection of the parameter \( \alpha \). Extensive experiments conducted by \cite{Ma2024evolution, Ni2019community} have demonstrated that setting \( \alpha = 0.5 \) yields the best results for graph community detection. We will use \( \alpha = 0.5 \) as the parameter setting for our experiments to maintain consistency. All experiments are conducted five times with different random seeds, and the average performance scores are reported. For traditional community discovery algorithms applied to graphs, we first transform the hypergraphs into graphs using clique expansion. In the case of TriCL and S$^2$-HHGR, their original designs are preserved, leveraging additional node features. In contrast, our proposed approach rely exclusively on the structural information of the hypergraph instead of its clique expansion. All baseline scores are sourced from Lee and Shin \cite{lee2023m} and Hacquard \cite{Hacquard2024hypergraph}. We conduct all experiments on a computing platform with an NVIDIA GPU (24GB memory) and a 16-core Intel CPU ( i9-12900KF), defining run failures as either out-of-memory (OOM) errors or exceeding a 48-hour runtime limit on this device.

\subsection{Experimental Results}
\subsubsection{Experimental Results for the Synthetic Networks.}
To rigorously evaluate the community detection performance, we leverage three series of synthetic networks with embedded community structures. The Normalized Mutual Information (NMI) is adopted as the evaluation metric, quantifying the similarity between detected communities and the ground-truth. Each network within the series is sequentially input into the HyperRCD and comparison methods, with the NMI results visualized in Figure \ref{result_D1}–\ref{result_D3}.

As illustrated in Figure \ref{result_D1}, when analyzing the relationship between community detection performance and average node degree, classical algorithms (e.g., Louvain, Girvan-Newman) exhibit divergent trends: some show marginal improvements, while others degrade or fluctuate notably with increasing average degree. In stark contrast, HyperRCD demonstrates exceptional stability, maintaining a consistently higher NMI across all tested average degree values. This stems from HyperRCD’s innovative hypergraph representation mechanism, which explicitly models higher-order node interactions through hyperedges. By capturing multi-wise relationships rather than pairwise connections, HyperRCD inherently adapts to diverse hyperedge configurations—this structural modeling innovation enables the algorithm to mitigate sensitivity to node connectivity density changes, thus achieving robust generalization in community identification where traditional methods falter.
\begin{figure}[htbp]
	\centering
	\includegraphics[width=0.65\linewidth]{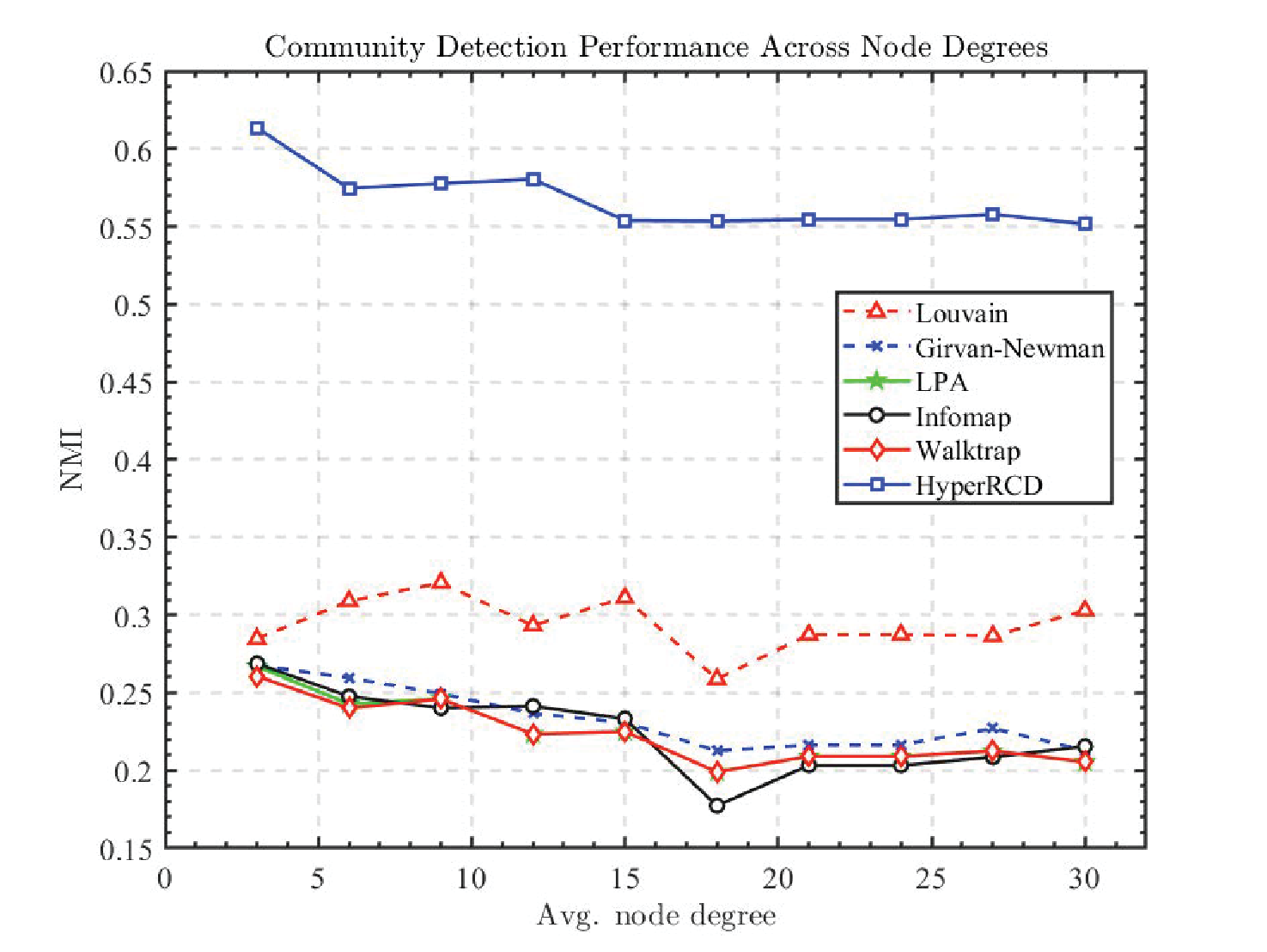}
	\caption{NMI performance on synthetic network D1}
	\label{result_D1}
\end{figure}
Figure \ref{result_D2} illustrates community detection performance across varying intra-community connectivity strengths. Although all methods benefit from increased connectivity, HyperRCD consistently surpasses the baseline approaches over the entire range. Notably, under lower connectivity conditions (0.15–0.3), certain traditional algorithms yield NMI values approaching 0. This occurs because methods like LPA tend to converge to trivial solutions—grouping all nodes into a single giant community due to their label update strategy. In contrast, HyperRCD maintains relatively high NMI scores, underscoring its robustness in resolving ambiguous community boundaries. This resilience is attributed to its advanced hypergraph modeling, which adeptly captures subtle structural dependencies even in sparse or weakly defined community contexts.
\begin{figure}[htbp]
	\centering
	\includegraphics[width=0.65\linewidth]{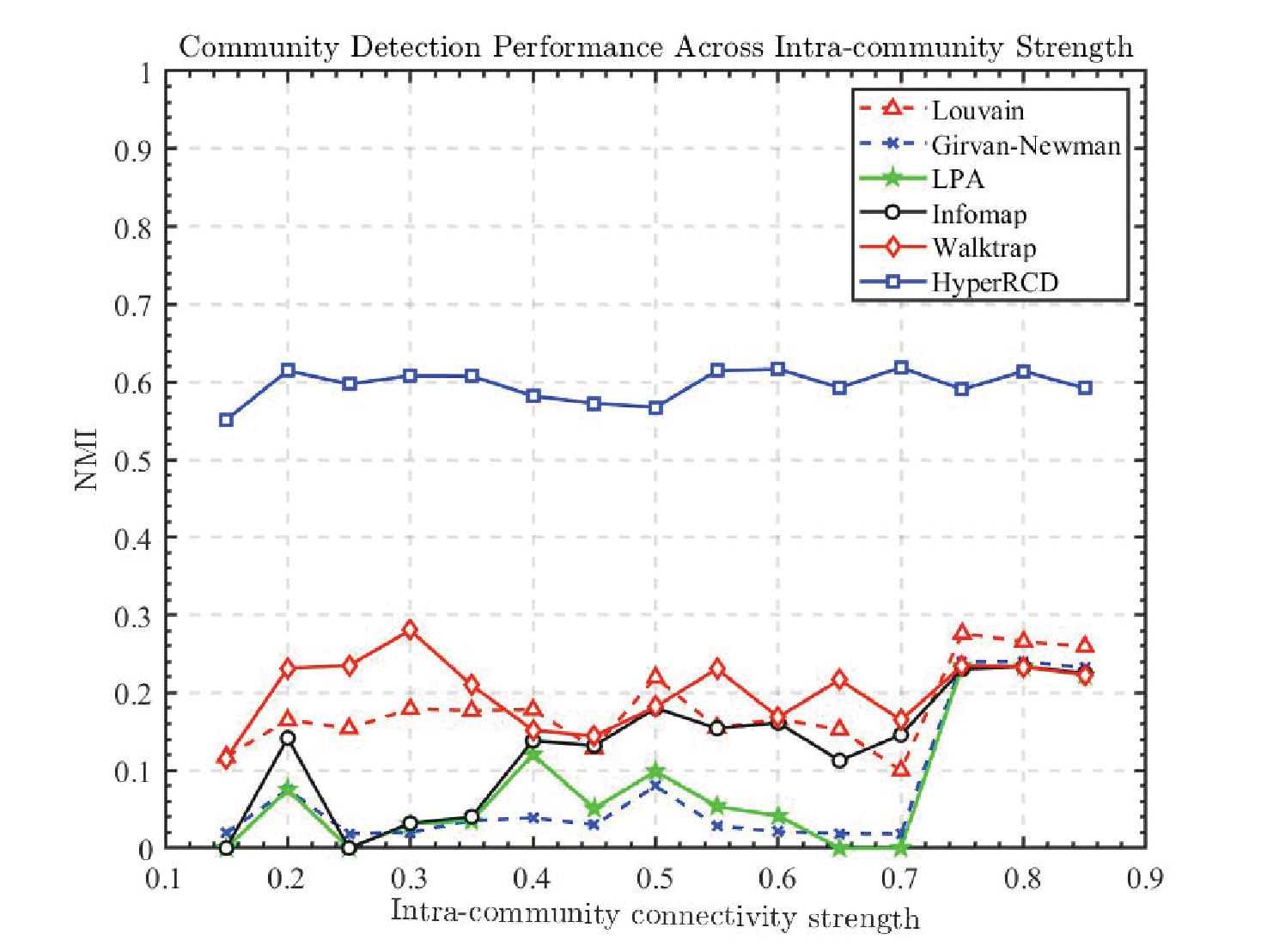}
	\caption{NMI performance on synthetic network D2}
	\label{result_D2}
\end{figure}
Figure \ref{result_D3} further validates the scalability of HyperRCD. As the network size grows (i.e., the number of nodes increases), most baseline algorithms suffer significant performance degradation, with some NMI values dropping below 0.2. Conversely, HyperRCD maintains a stable NMI above 0.6, showcasing its capability to handle large-scale networks without compromising detection accuracy. This scalability stems from HyperRCD’s hypergraph-based formulation and optimized computational strategy, which efficiently manage the complexity escalation associated with network expansion.
\begin{figure}[H]
	\centering
	\includegraphics[width=0.65\linewidth]{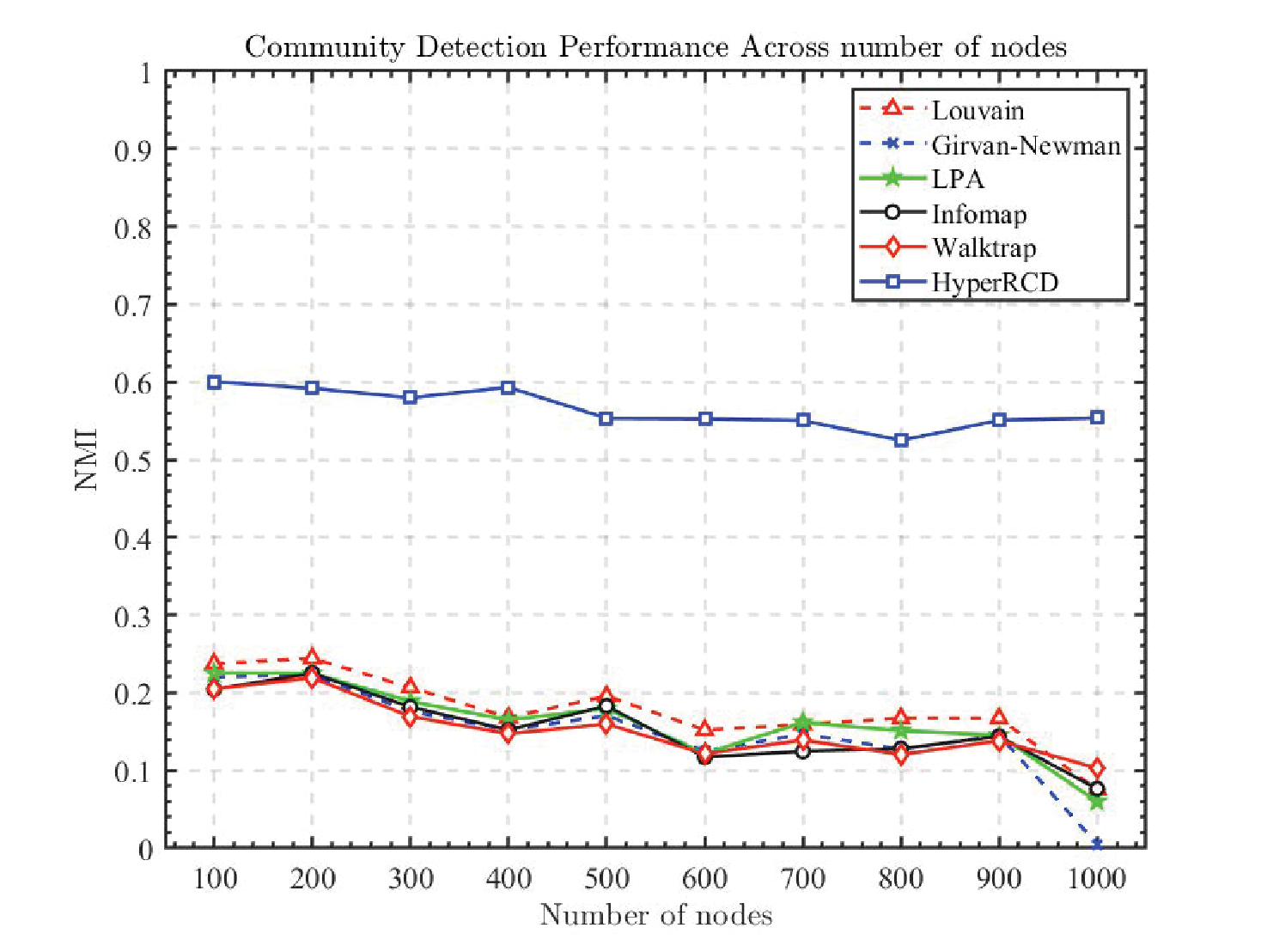}
	\caption{NMI performance on synthetic network D3}
	\label{result_D3}
\end{figure}
Overall, these experiments demonstrate that HyperRCD outperforms widely used community detection algorithms across all three dimensions: it scales more gracefully with network size, remains robust in the presence of ambiguous community boundaries, and generalizes effectively to various hyperedge distributions.

We visualize HyperRCD on two synthetic datasets with different scales (Figures \ref{fig:visualization_syn1} and \ref{fig:visualization_syn2}), comparing against both the original structures and Louvain-detected results. Specifically, Figure \ref{fig:visualization_syn1} displays the visualization of a D1-series hypergraph with 100 nodes, 3 communities, avg. degree$=15$, and intra-community connectivity strength$=0.85$; Figure \ref{fig:visualization_syn2} displays the visualization of a D3-series hypergraph with 500 nodes, 10 communities, avg. degree$=10$, and intra-community connectivity strength$=0.85$. For clearer community visualization, we display clique-expanded graphs of the hypergraphs, where nodes sharing labels share colors, intra-community edges match node colors and inter-community edges use gradient colors. It is clear that HyperRCD outperforms Louvain in preserving intra-community connectivity and achieving label-consistent partitions for the synthetic datasets with different scales.
\begin{figure}[htbp] 
	\centering
	\subfigure[Original clique-expanded graph ]{ 
		\includegraphics[width=0.3\textwidth]{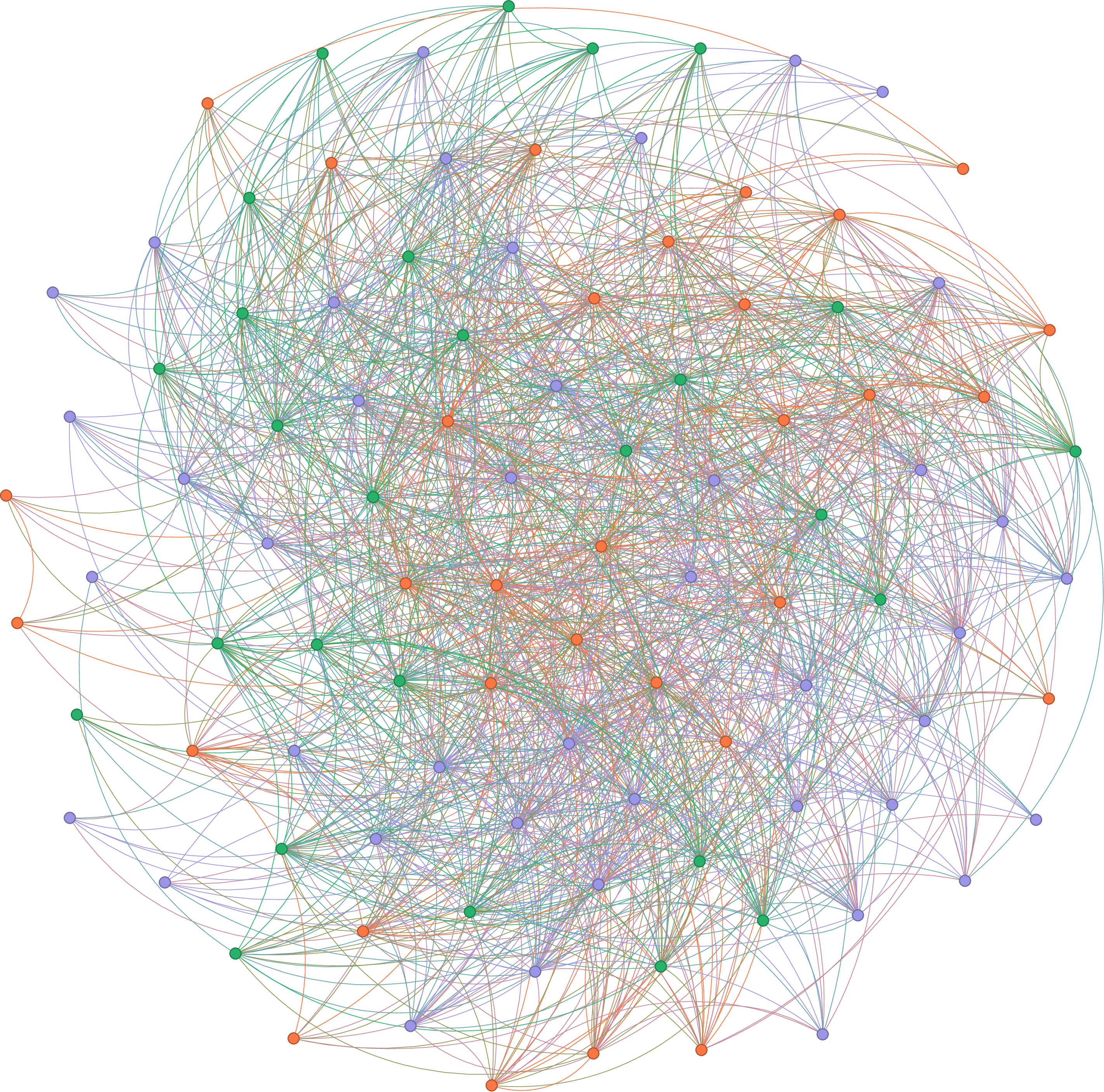}
	}
	\hfill 
	\subfigure[After HyperRCD processing]{ 
		\includegraphics[width=0.3\textwidth]{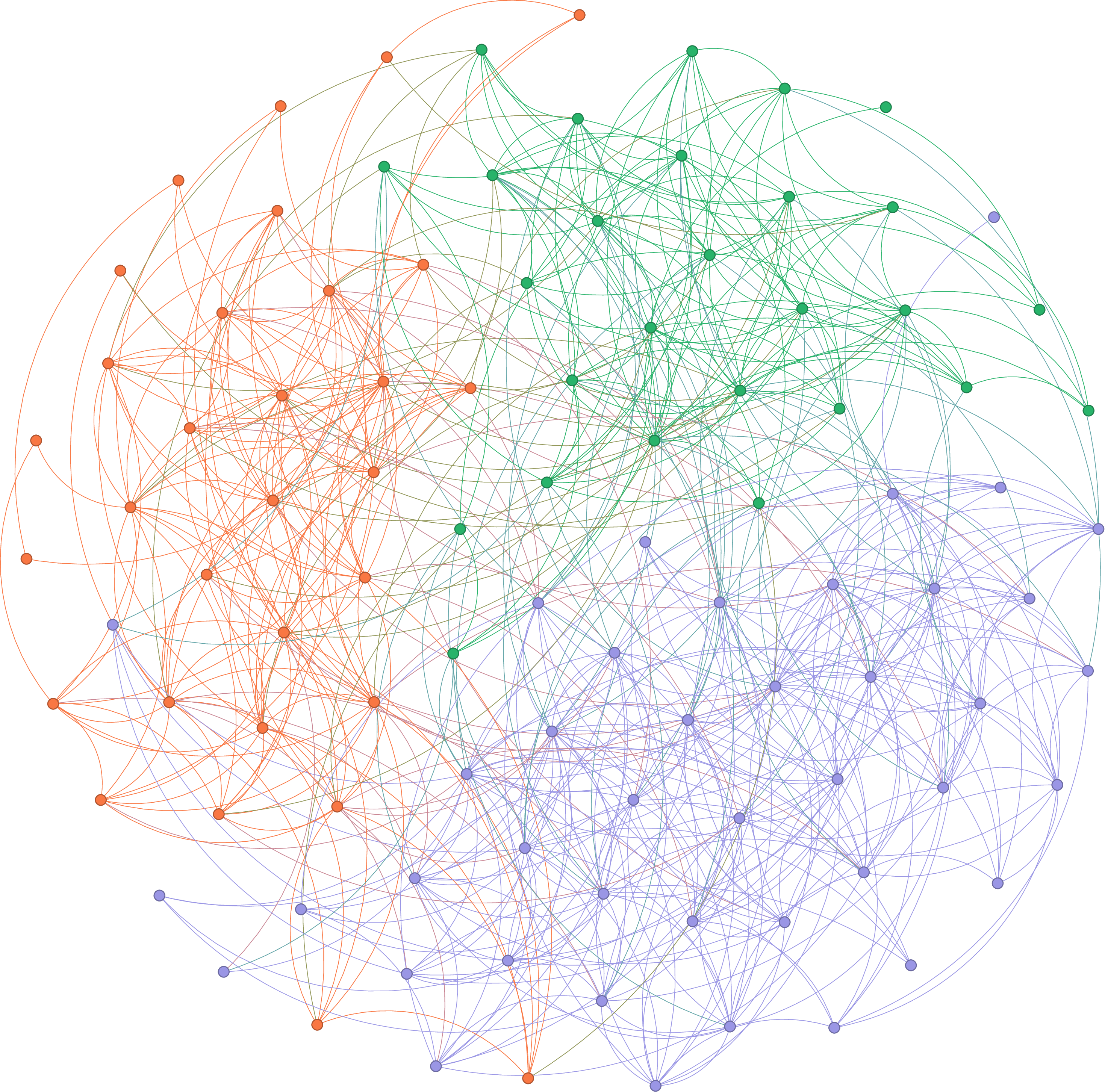}
	}
	\hfill
	\subfigure[After Louvain processing]{ 
		\includegraphics[width=0.3\textwidth]{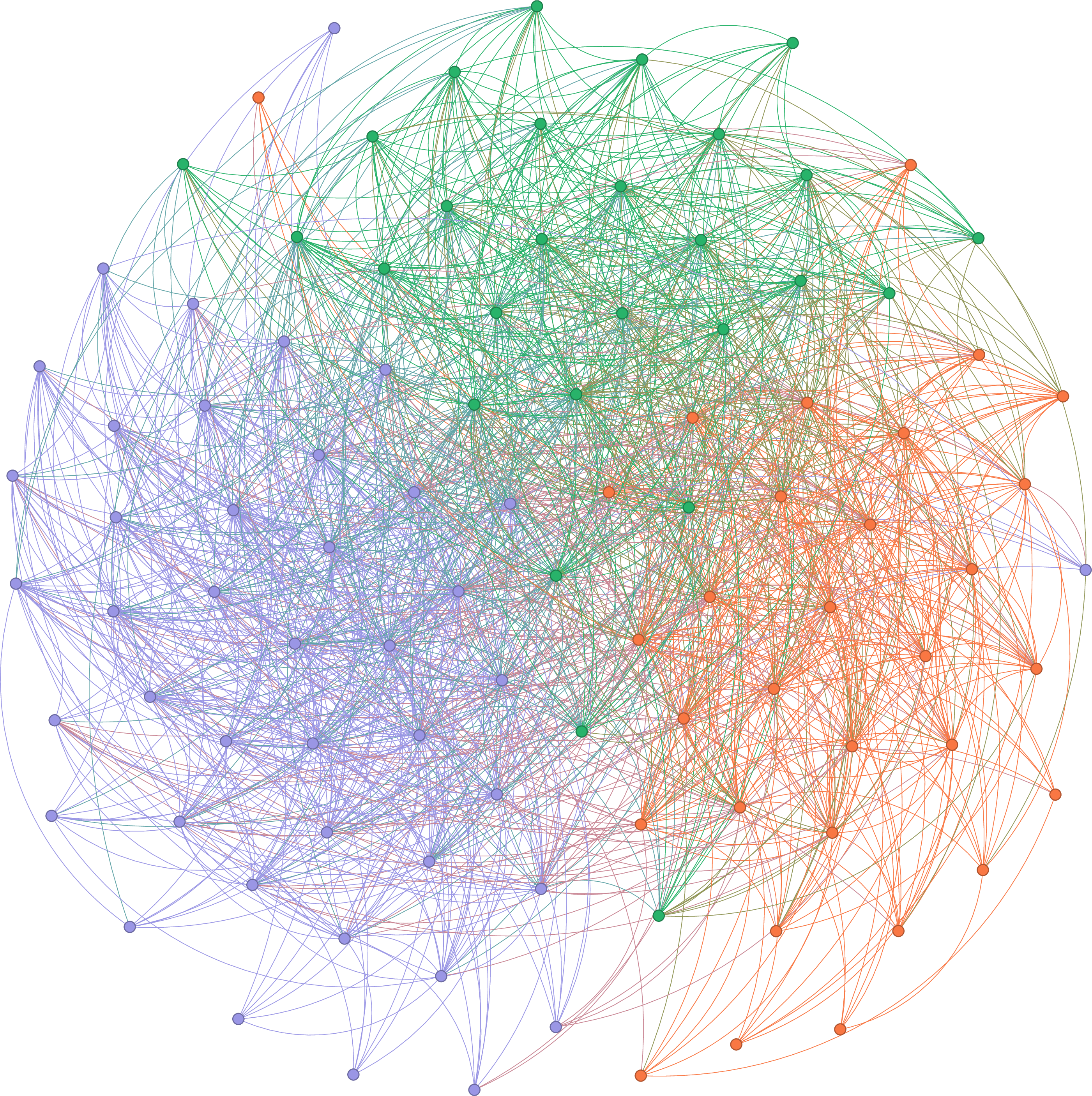}
	}
	\caption{HyperRCD Community Detection on D1: Original vs. HyperRCD vs. Louvain} 
	\label{fig:visualization_syn1} 
\end{figure}

\begin{figure}[htbp] 
	\centering
	\subfigure[Original clique-expanded graph]{ 
		\includegraphics[width=0.3\textwidth]{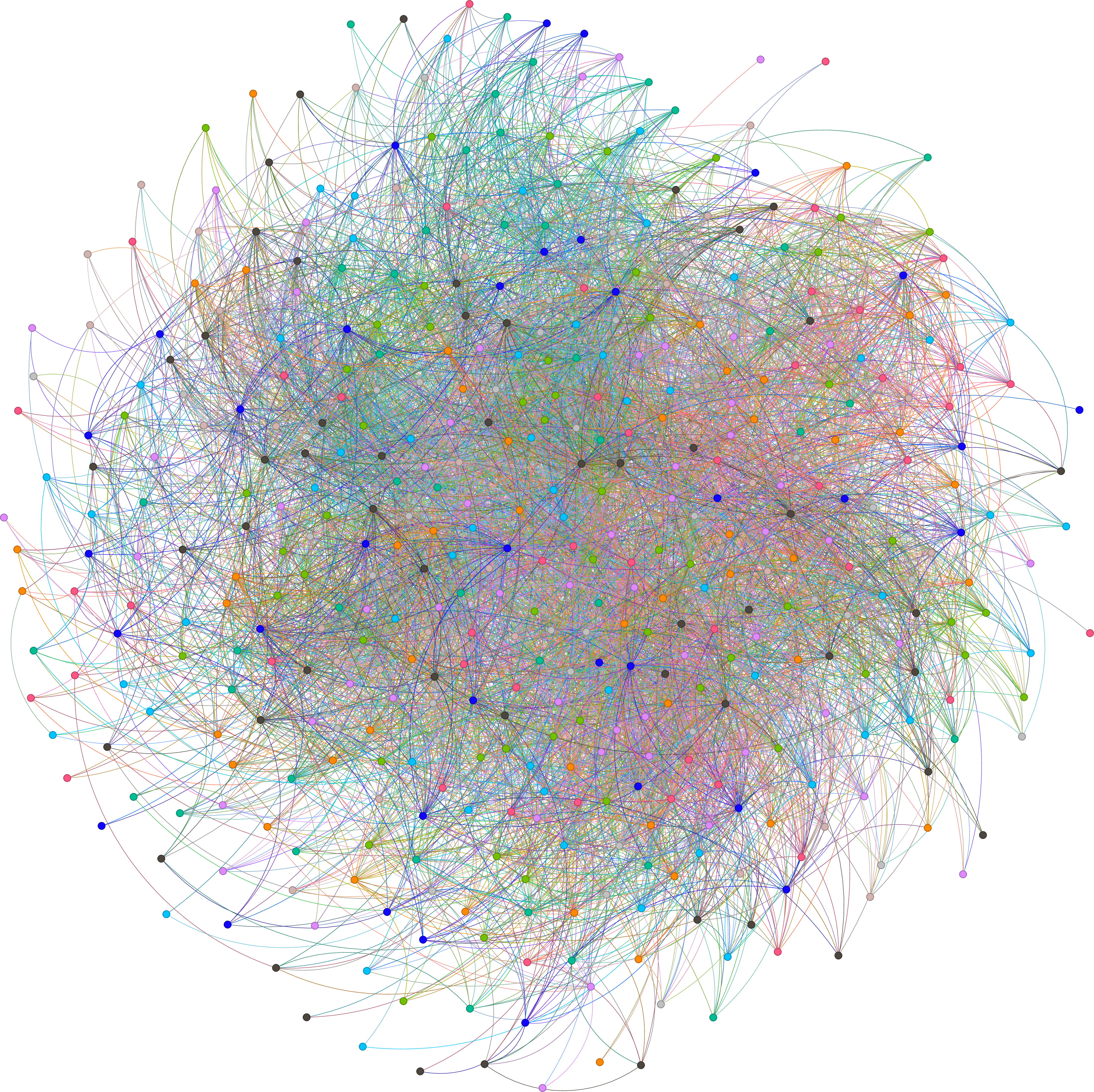}
	}
	\hfill 
	\subfigure[After HyperRCD processing]{ 
		\includegraphics[width=0.28\textwidth]{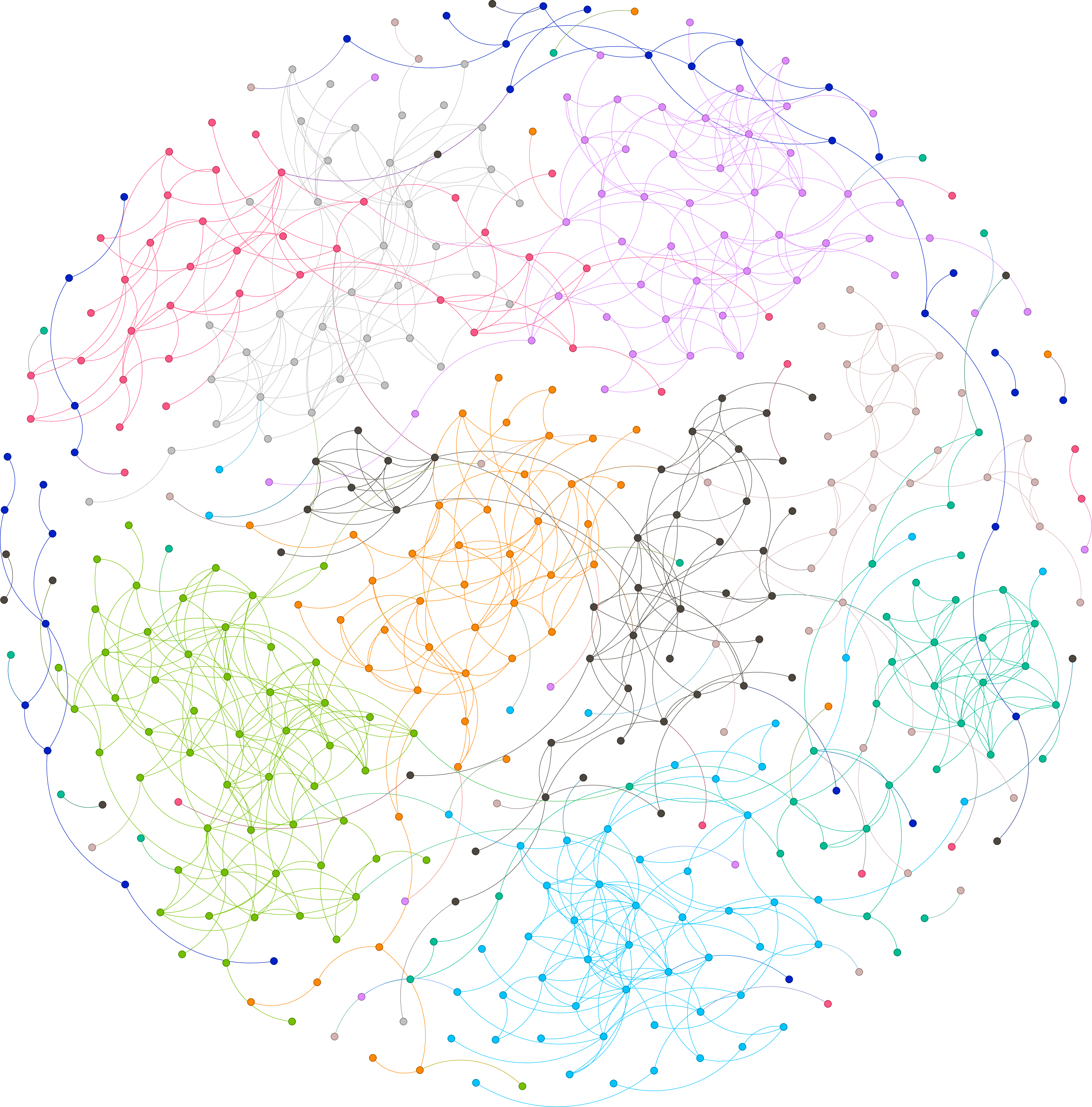}
	}
	\hfill
	\subfigure[After Louvain processing]{ 
		\includegraphics[width=0.28\textwidth]{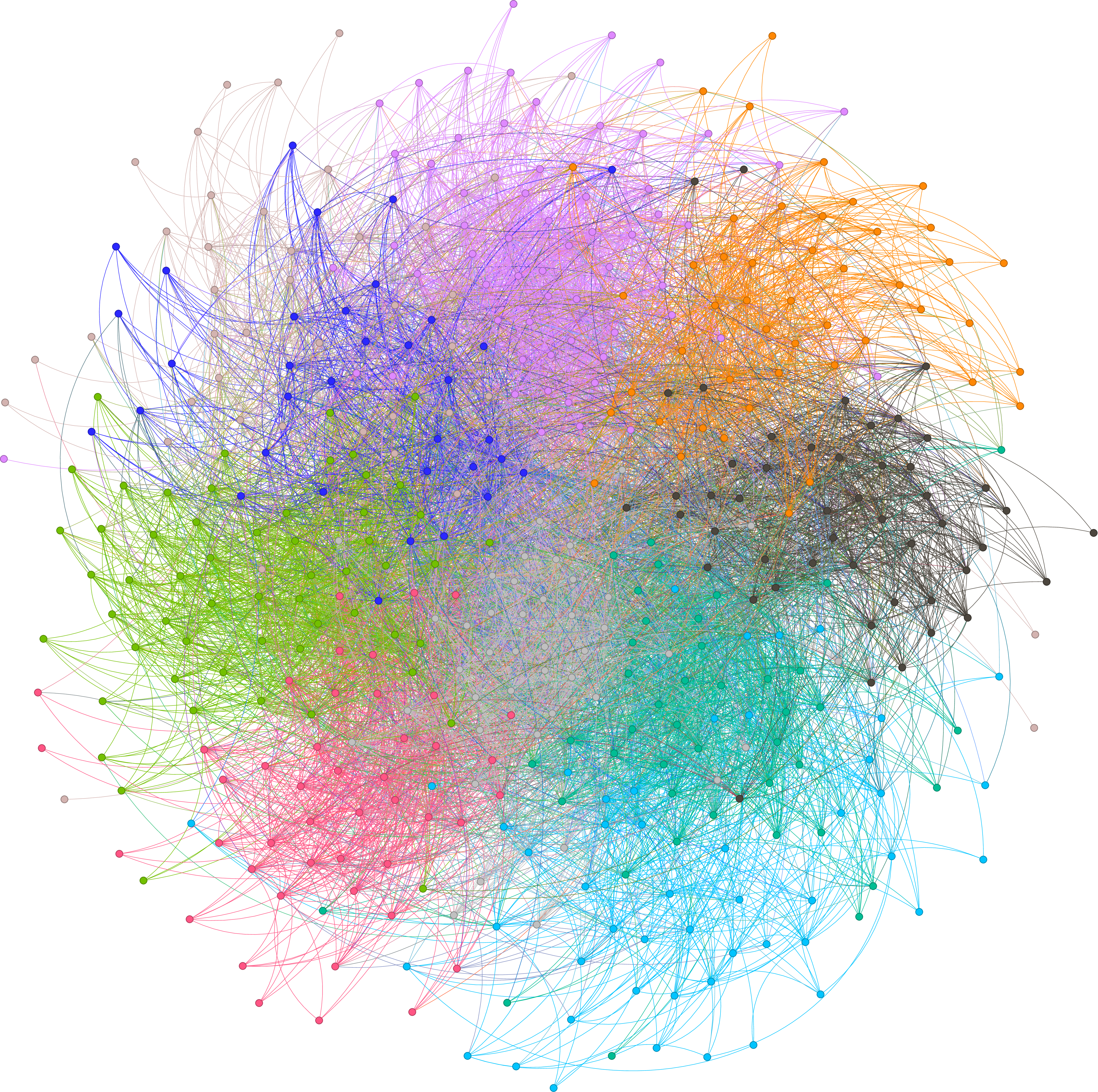}
	}
	\caption{HyperRCD Community Detection on D3: Original vs. HyperRCD vs. Louvain} 
	\label{fig:visualization_syn2} 
\end{figure}

\subsubsection{Experimental Results for the Real-World Networks}

\begin{table}[h]
	\caption{\label{tab:3}Performance of various algorithms on real-world networks}
	\centering
	\begin{tabular}{cccccccc}
		\toprule
		Methods\textbackslash{}Networks & Zoo & Mushroom & Cora-C & Citeseer & Pubmed & Cora-A & NTU2012\\
		\midrule
		Node2Vec & 0.115 & 0.016 & 0.391 & 0.245 & 0.231 & 0.160 & 0.783\\
		DGI & 0.130 & OOM & \textbf{0.548} & 0.401 & \textbf{0.304} & 0.452 & 0.796\\
		GRACE & 0.073 & OOM & 0.444 & 0.333 & 0.167 & 0.379 & 0.746\\
		S$^2$-HHGR & 0.909 & 0.186 & 0.510 & 0.411 & 0.277 & 0.454 & 0.827\\
		TriCL & 0.912 & 0.038 & 0.545 & \textbf{0.441} & 0.300 & \textbf{0.498} & 0.832\\
		Modularity & 0.777 & \textbf{0.434} & 0.450 & 0.338 & 0.250 & 0.334 & 0.745\\
		N - Ricci & 0.962 & OOT & 0.458 & 0.388 & 0.278 & 0.394 & 0.823\\
		\textbf{HyperRCD} & \textbf{0.981} & OOT & 0.484 & 0.412 & 0.288 & 0.423 & \textbf{0.836}\\
		\bottomrule
	\end{tabular}
\end{table}

Table \ref{tab:3} summarizes the empirical performance of all methods, and the best results are marked in bold. HyperRCD achieves state-of-the-art results in 2 of 7 datasets, including the Zoo dataset, where it outperforms N-Ricci, which is also a curvature-driven method, by $1. 9\%$ in normalized mutual information (NMI: 0.981 vs. 0.962), validating its robustness in detecting clusters. On NTU2012, HyperRCD achieves an NMI of 0.836, surpassing TriCL (0.832) and $S^2$-HHGR (0.827), underscoring its ability to handle heterogeneous hyperedge structures. In Table \ref{tab:3}, OOM (out-of-memory) and OOT (out-of-time) denote computational failures due to memory exhaustion and time limit exceedance, respectively.

For intuitive evaluation of HyperRCD's effectiveness for real-world networks, we present visualizations of both the original data and community detection results in Figure \ref{fig:visualization_zoo}. Since the hyperedges may obscure the visualization of vertex communities, we visualize the clique-expanded graphs rather than raw hypergraphs. Due to excessive density in other datasets (with large vertex/hyperedge counts), we limit visualization to the Zoo dataset for clarity. The visualization (Figure \ref{fig:visualization_zoo}) assigns identical colors to nodes within the same ground-truth community, intra-community edges match node colors and inter-community edges use gradient colors. HyperRCD’s Ricci-flow-driven partitioning maintains dense intra-comminity connections and exhibits near-perfect alignment with the ground-truth community labels.
\begin{figure}[htbp] 
	\centering
	\subfigure[Original clique-expanded graph]{ 
		\includegraphics[width=0.45\textwidth]{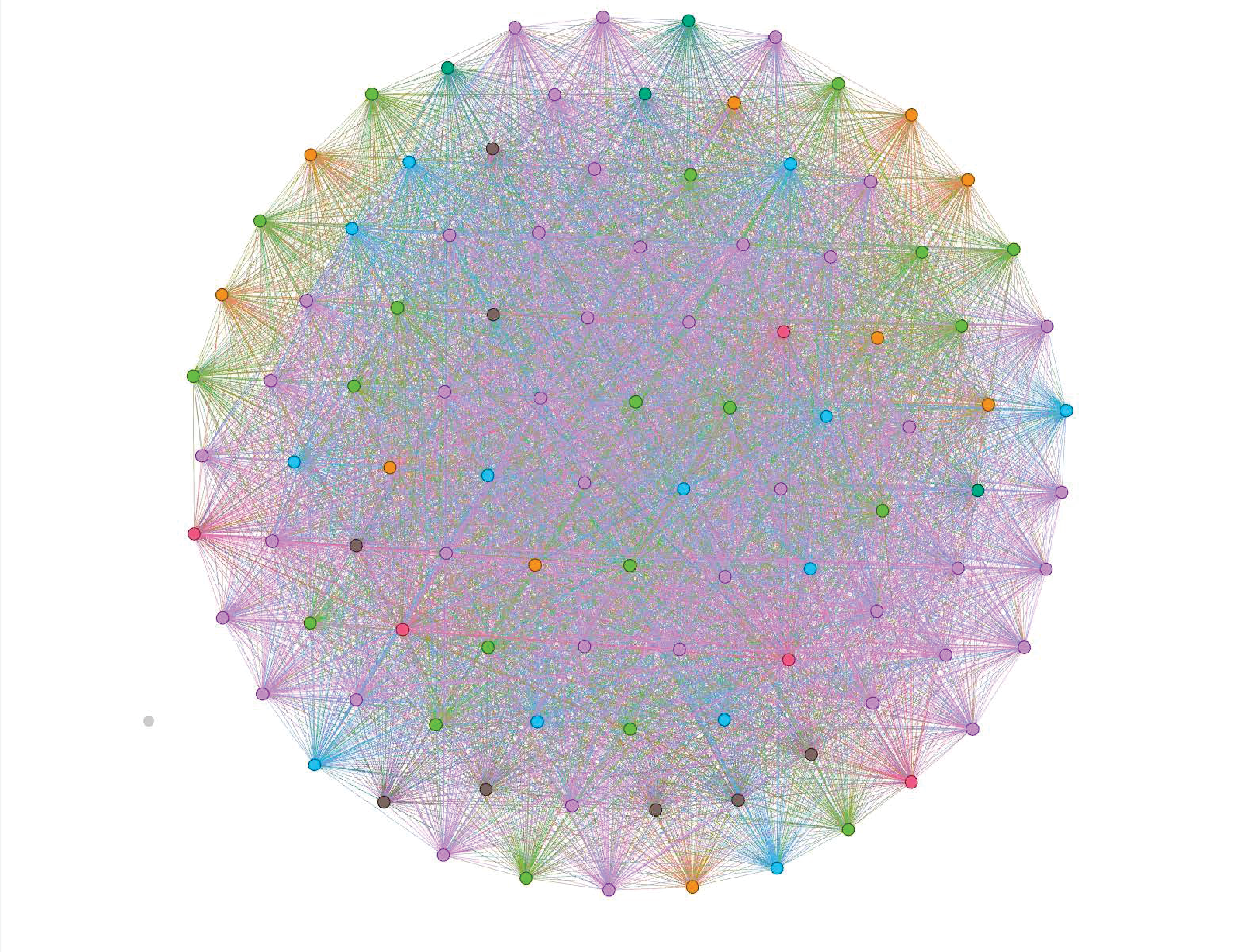} 
	}
	\hfill 
	\subfigure[After HyperRCD processing]{ 
		\includegraphics[width=0.45\textwidth]{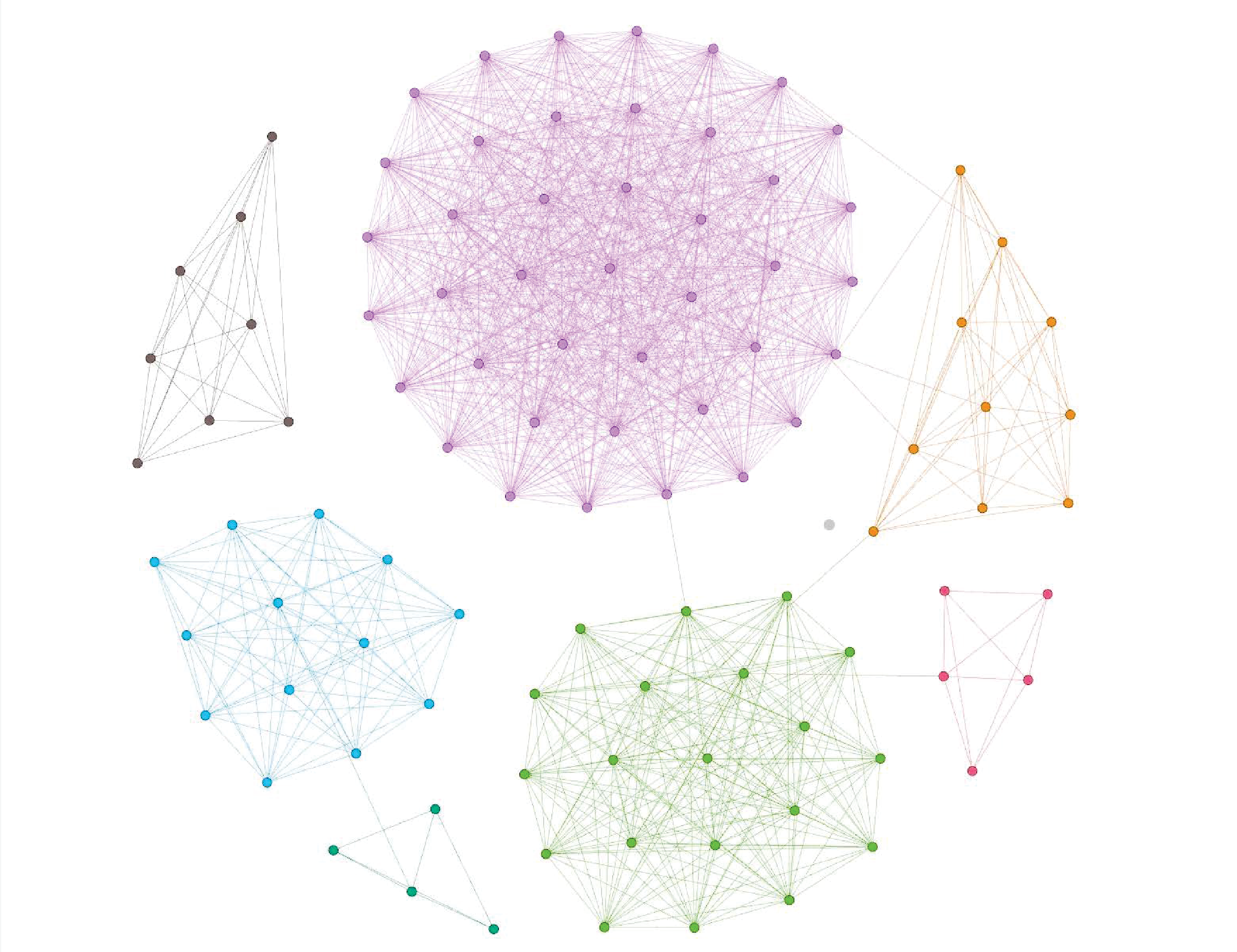}
	}
	\caption{HyperRCD Community Detection on Zoo Dataset: Original vs. Processed} 
	\label{fig:visualization_zoo} 
\end{figure}

Compared against traditional graph embedding approaches (Node2Vec, DGI, GRACE), HyperRCD systematically outperforms these methods across all datasets. For example, on Citeseer, the HyperRCD NMI (0.412) exceeds the Node2Vec score (0.245) by $68\%$, confirming the superiority of hypergraph modeling in capturing higher-order interactions.  

Against advanced hypergraph-specific methods ($S^2$-HHGR, TriCL), HyperRCD leads on Zoo and NTU2012 while maintaining competitive scores on Cora-C (0.484) and Citeseer (0.412). In Pubmed, HyperRCD trails TriCL by a marginal 0.012 NMI, suggesting the potential for improvement in sparse hypergraphs.  

Despite achieving high accuracy on most datasets, HyperRCD fails to process the Mushroom dataset (8,124 nodes) within the 48-hour time limit due to its \(O(ED^3)\) complexity, primarily driven by solving linear programming problems for Wasserstein distance calculations in hypergraphs with large cardinality of hyperedges (e.g., Mushroom’s average hyperedge size of 136.3 nodes). This limitation aligns with theoretical expectations and mirrors challenges faced by other methods: DGI and GRACE encounter memory issues (OOM), while N-Ricci exceeds the time limit (OOT). Although Node2Vec and TriCL complete execution, their NMI scores (0.016 and 0.038, respectively) remain significantly lower than HyperRCD’s performance on other datasets, underscoring a trade-off between scalability and accuracy in hypergraph community detection. These results emphasize the need for algorithmic optimizations, such as approximate Wasserstein distance approximations or parallel computing frameworks, to address computational bottlenecks in large-scale hypergraphs.

In summary, the proposed method HyperRCD demonstrates competitive accuracy across diverse real-world datasets, consistently outperforming traditional graph-based embedding techniques while achieving parity with state-of-the-art hypergraph neural embedding methods on select benchmarks. Specifically, it surpasses Node2Vec, DGI, and GRACE by significant margins and matches TriCL’s performance on NTU2012. These results validate the approach’s ability to leverage hypergraph curvature dynamics for effective community detection, even in real complex network structures.


\section{Conclusion}
We introduce an Ollivier-type Ricci curvature for hypergraphs with higher-order interactions, and establish a curvature-driven Ricci flow for weight evolution, proving its long-time existence. Building on this theoretical foundation, we propose HyperRCD, a novel algorithm for hypergraph community detection, which processes hypergraphs in their native form without graph reduction (clique/star expansion), preserving original higher-order structures. Extensive experiments are conducted on both synthetic and real-world datasets spanning graphs and hypergraphs with diverse structural characteristics. Comparative results demonstrate that HyperRCD achieves SOTA or highly competitive performance across benchmarks, exhibiting remarkable robustness and accuracy on various hypergraph datasets.

\begin{acknowledgments}
This research is supported by National Natural Science Foundation of China (No. 12271039) and the Open Project Program (No. K202303) of Key Laboratory of Mathematics and Complex Systems, Beijing Normal University.
\end{acknowledgments}




\section*{References}

\bibliography{mybib-hypergraph}

\end{document}